\documentclass[11pt]{llncs}

\usepackage{geometry}

\usepackage{amssymb}
\usepackage{graphicx}
\usepackage{algorithm}
\usepackage{algorithmic}
\usepackage{tikz} 
\usetikzlibrary{arrows}
\usepackage{enumerate}

\makeatletter
\newcommand{\newalgname}[1]{
  \renewcommand{\ALG@name}{#1}
}

\newalgname{Figure}
\renewcommand{\listalgorithmname}{List the \ALG@name}
\makeatother

\newtheorem{thm}{Theorem}
\newtheorem{inv}{Invariant}
\newtheorem{lem}{Lemma}

\usepackage{url}
\urldef{\mailsa}\path|{alfred.hofmann, ursula.barth, ingrid.haas, frank.holzwarth,|
\urldef{\mailsb}\path|anna.kramer, leonie.kunz, christine.reiss, nicole.sator,|
\urldef{\mailsc}\path|erika.siebert-cole, peter.strasser, lncs}@springer.com|    
\newcommand{\keywords}[1]{\par\addvspace\baselineskip
\noindent\keywordname\enspace\ignorespaces#1}

\begin{document}

\mainmatter  % start of an individual contribution

\title{Group Mutual Exclusion in Linear Time and Space}

\titlerunning{Group Mutual Exclusion in Linear Time and Space}

\author{Yuan He\inst{1}%
\and Krishnan Gopalakrishnan\inst{2} \and Eli Gafni\inst{1}} 
\authorrunning{Group Mutual Exclusion in Linear Time and Space}
\institute{
Department of Computer Science \\
University of California at Los Angeles \\
Los Angeles, CA-90095
\and
Department of Computer Science \\
East Carolina University \\ Greenville, NC-27858}


\maketitle

\begin{abstract}
We present two algorithms for the {\bf Group Mutual Exclusion (GME) Problem} that satisfy the properties of {\em Mutual Exclusion, Starvation Freedom, Bounded Exit, Concurrent Entry and First Come First Served}.
Both our algorithms use only simple read and write instructions, have $O(N)$ Shared Space complexity and $O(N)$ Remote Memory Reference (RMR) complexity in the Cache Coherency (CC) model.
Our first algorithm is developed by generalizing the well-known Lamport's Bakery Algorithm for the classical mutual exclusion problem, while preserving its simplicity and elegance.
However, it uses unbounded shared registers.
Our second algorithm uses only bounded registers and is developed by generalizing Taubenfeld's Black and White Bakery Algorithm to solve the classical mutual exclusion problem using only bounded shared registers.
We show that contrary to common perception our algorithms are the first to achieve these properties with these combination of complexities.

\end{abstract}
\begin{enumerate}
\item Contact author's email address is gopal@ecu.edu and phone number is (252)-328-9688. His address is Suite C-114, Science and Technology Building, Department of Computer Science, East Carolina University, Greenville, NC-27858, USA.

\item This paper is a regular submission.
\item This paper is eligible to be considered for the best student paper award as the first author is a full-time student.
\end{enumerate}
\newpage

\setcounter{page}{1}
\pagestyle{plain}

\section{Introduction}
\label{intro}

{\em Mutual Exclusion} is a classical problem in distributed computing introduced by Dijkstra in 1965 \cite{Dij65}.
The {\em Group Mutual Exclusion (GME)} problem, introduced by Joung in 1998 \cite{Jou00}, is a natural generalization of the classical mutual exclusion problem.

In Group Mutual Exclusion problem, processes repeatedly cycle through four sections of code viz., {\em Remainder Section, Entry Section, Critical Section (CS)} and {\em Exit Section}, in that order.
An execution of the last three sections will be called an {\em invocation}.
A process picks a session number when it leaves the remainder section and this session number can be different in each invocation.
A process is said to be an {\em active} process, if it is in one of its invocations.
Two active processes are in {\em conflict} if their session numbers are different.  
Unlike the classical mutual exclusion problem, multiple processes are allowed to be in the critical section at the same time, provided they are not conflicting. In fact, in the presence of active processes all of which are mutually non-conflicting, entry into the critical section takes bounded number of process steps.
Formally, the problem consists of designing code for the entry section and the exit section such that the following four properties are satisfied.

\begin{description}
	\item[P1] {\bf Mutual Exclusion:} No two conflicting processes can be in the critical section at the same time.
	\item[P2] {\bf Starvation Freedom:} If no process stays in the critical section forever, then any process that enters the entry section eventually
	enters the critical section.
	\item[P3] {\bf Bounded Exit:} After entering the exit section, a process is guaranteed to leave it within a bounded number of its own steps.
	\item[P4] {\bf Concurrent Entry:} In the absence of conflicting processes, a process in the entry section should be guaranteed to enter the critical section within a bounded number of its own steps.
\end{description}

The {\em Concurrent Entry} property is crucial to the GME problem.
It was stated informally by Joung \cite{Jou00} and then was later formalized by Hadzilacos \cite{Hadz01}.
The intent of this property is to ensure concurrency: active processes that request the same session, in the absence of conflicting processes, should be allowed to enter the CS without unnecessary synchronization among themselves.

In this paper, we will require that the algorithm satisfies the {\em First Come First Served (FCFS)} property in addition to the above four properties.
The standard way to formalize this is to split the entry section to two sections viz., {\em Doorway Section} and {\em Waiting Room Section}.
The doorway section of the code is free of ``wait'' statements, i.e., it can be completed by a process in a bounded number of its own steps.
The waiting room section is where the actual synchronization with other conflicting processes occurs and may entail indefinite waiting.
The notion of doorway was originally introduced by Lamport \cite{Lamp74} in the context of classical mutual exclusion problem.
We would say that process $P_{i}$ {\em doorway precedes} process $P_{j}$, if $P_{i}$ completes the doorway before $P_{j}$ enters the doorway.
Now, the {\em FCFS} property can be formally stated as given below.

\begin{description}
	\item[P5] {\bf FCFS:} If process $P_{i}$ doorway precedes process $P_{j}$ and the two processes request different sessions, then process $P_{j}$ does not enter the critical section before process $P_{i}$.
\end{description}

\subsection{Model}
\label{model}

We consider a system consisting of $N$ processes, named $P_{1}, P_{2} \ldots P_{N}$ and a set of shared variables.
Each process also has its own private variables.
Processes can communicate only by writing into and reading from shared variables.
An execution is modeled as a sequence of process steps.
In each step, a process performs some local computation or reads from a shared variable or writes into a shared variable. 
We assume that these steps of reading or writing are atomic.
The processes take steps {\em asynchronously}. 
Specifically, this means that an unbounded number of steps of some other process could be performed in between two consecutive steps of a process.
We assume that our processes are {\em live}; this means that if a process is active it will eventually execute its next step.

We allow only simple read and write operations on shared variables. We assume that these
read and write operations are atomic. However, we do not assume that processes have access to more powerful synchronization operations such as atomic test-and-set, compare-and-swap etc.

In this paper, we will be working exclusively under the Cache-Coherent (CC) model, a model of practical significance considering that virtually
every modern multi-processor is cache-coherent.  
In the CC model, all shared variables are stored in a global
memory module that is not associated with any particular processor.
Each processor also has a local cache and some
hardware protocol ensures cache coherence
i.e., all copies of the same variable in different local caches
are consistent. Every time a process reads a shared variable,
it does so using a local (cached) copy of the variable. A
local copy of the variable may not be valid, if the process
has never read the variable before or if some process overwrote
it in the global memory module. Whenever
upon reading a cached variable, the process is informed by the system
that its value is invalid, it makes a
remote memory reference and migrates the variable to its
local cache. We assume that once a shared variable is brought
into a process's local cache it remains there indefinitely. Also,
every time a process writes a shared variable, the process
writes the variable in the global memory module, which involves
a remote memory reference. Note that, this action 
invalidates all cached copies of that variable.

   Our time complexity is the number of Remote Memory Reference
(RMR) steps. This is because remote memory
references are the most time consuming operations as they
involve interconnect traversal. Our space complexity is 
the total amount of shared space a solution entails.
We do not count the private variables when measuring the space
complexity.  We consider two different types of shared variables  situations,
one in which the shared variables are bounded registers
and one in which the shared variables are unbounded registers.

\subsection{Our Contribution and Related Work}

   Joung's original algorithm for the GME problem
satisfies the four basic properties. It does not
satisfy the FCFS property. Moreover, it has unbounded
RMR complexity.
Hadzilacos (see \cite{Hadz01}) gave the first solution
for the GME problem that has the FCFS property.
His algorithm can be thought of
as a modular composition of two independent algorithms,
one, ``the FCFS algorithm'' provides FCFS property
(but not necessarily guarantee mutual exclusion) and
the other, ``the ME algorithm'' provides mutual exclusion property
(but not necessarily FCFS). This algorithm has
shared space complexity of $\Theta(N^2)$. It was (mistakenly in hindsight)
claimed that the algorithm has RMR complexity
of $O(N)$ in the CC model. The algorithm was using only bounded 
shared variables and simple read and write operations. It was left as an open problem to develop
a solution (satisfying P1 through P5)
for the GME problem that runs in linear time and
space using only bounded shared variables.

 Subsequently, \cite{JPT03}, Jayanti et al.
presented an algorithm presumed to be of linear time and space,
solving the challenge of Hadzilacos.  Jayanti et al. retained the idea of modular composition
and also the ``ME algorithm'' that Hadzilacos used. They
came up with a clever modification to the ``FCFS algorithm''
of Hadzilacos to reduce the space complexity to $\Theta(N)$.
Unfortunately, though without writing it explicitly,
they inherited from Hadzilacos the mistaken impression 
that the algorithm they use has $O(N)$ time complexity.
%Although they did not explicity claim so, their algorithm
%is considered to have inherited the linear time from Hadzilacos.
\footnote{
For example, the recent paper by Bhatt and Huang \cite{BH10} explicitly
states that the RMR complexity of the algorithm by Jayanti
et al. is $O(N)$.
}

 Both works use a slightly modified version of a classical mutual exclusion algorithm developed
independently by Burns \cite{Burns81} and  Lamport \cite{Lamp86}
as the ``ME algorithm''. This is an elegant algorithm
that uses just one bit of shared space per process.
In Section~\ref{flaw}, we show that this algorithm
actually has an intricate structure and has the
worst case RMR complexity of $\Omega(N^2)$. It follows
that algorithms of both Hadzilacos and Jayanti et al.
for solving GME are of RMR complexity $\Omega(N^2)$. Hence,
the challenge posed by Hadzilacos has, yet, not been met.
Our observation, and part of our contribution, that the challenge is still on, initiated this paper.

  Our first algorithm, presented in Section~\ref{simple},
is a generalization of the classic Lamport's Bakery
Algorithm to solve the GME problem while maintaining its
simplicity and elegance. It uses unbounded registers 
to solve the GME problem (satisfying P1 through P5)
and runs in linear time and space.

Takamura and Igarashi also made an attempt in \cite{TI03} to generalize Lamport’s Bakery Algorithm to solve the GME problem.
They presented three different algorithms in that paper.
However, all of their algorithms satisfy neither the concurrent entry property
nor the FCFS property.

In 2004 (which is in the future of \cite{Hadz01} and \cite{JPT03}), 
Taubenfeld \cite{Tau04} came up with an elegant algorithm
called \emph{Black and White Bakery Algorithm} that
solves the classical mutual exclusion problem with
only bounded shared registers. His approach is a lot
simpler than a prior method that bounds
the registers of the Lamport's Bakery Algorithm developed by Jayanti et al. in \cite{JTFK01}. 
Our second algorithm,
presented in Section~\ref{bwgme} is a generalization
of the Black and White Bakery Algorithm to solve
the GME problem. 
Our algorithm satisfies the properties P1 through P5
and runs in linear time and space using bounded shared registers. 
Thus, our algorithm is the
first one to solve the open problem originally posed
by Hadzilacos, a decade and a half ago.

   We present a comparison of different GME algorithms in Table 1.
To make the comparison fair, we include only those algorithms that solve the problem using only simple
read/write instructions. The first row in the table describes the properties of Joung's original
algorithm for the GME problem. The next four rows compares the algorithms for the GME problem that uses
unbounded shared registers. The last four rows compares the algorithms for the GME problem that uses
bounded shared registers. 

\begin{table}
\begin{center}
\caption{Comparison of GME algorithms that use only simple read/write instructions}

\begin{tabular} {||l||c|c|c|c|c|c|c||} \hline
Algorithm & RMR & Space & P2? & P3? & P4? & P5? & Bounded? \\ \hline \hline

Joung  \cite{Jou00} &  $\infty$ &  O(N) & Yes & Yes & Yes & No & Yes  \\ \hline \hline
Takamura et. al. \cite{TI03} I & $\infty$ & O(N) & No & Yes & No & No & No \\ \hline
Takamura et. al. \cite{TI03} II & O(N) & O(N) & Yes & No & No & No & No \\ \hline
Takamura et. al. \cite{TI03} III & O(N) & O(N) & Yes & No & No & No & No \\ \hline
This work ({\bf GLB Algorithm}) & O(N) & O(N) & Yes & Yes &  Yes & Yes & No \\ \hline \hline
Keane \& Moir \cite{KM99} & O($\log N$) & O(N) & Yes & No & No & No & Yes \\ \hline
Hadzilacos  \cite{Hadz01} & O($N^{2}$) & O($N^{2}$) & Yes & Yes &  Yes & Yes & Yes \\ \hline
Jayanti et. al. \cite{JPT03} & O($N^{2}$) & O(N) & Yes & Yes & Yes & Yes & Yes \\ \hline
This work ({\bf BWBGME Algorithm}) & O(N) & O(N) & Yes & Yes & Yes &  Yes & Yes\\ \hline \hline
\end{tabular}
\end{center}
\end{table}

\section{Generalizing Lamport's Bakery Algorithm}
\label{simple}

In this section, we present a very simple algorithm
for the GME problem by generalizing Lamport's 
Bakery Algorithm (see \cite{Lamp74})
for the classical mutual exclusion problem.
This algorithm is based on the method commonly used in bakeries, 
in which a customer receives a token number upon entering the 
store and the holder of the lowest token number is the next 
one served.  We will refer to it as Generalized Lamport's 
Bakery (GLB) Algorithm.

The algorithm is presented in Figure \ref{GLBAlgo}.
It uses three shared variables.
The first one \emph{Session} is an integer array of size $N$ and \emph{Session}$[i]$ 
indicates the session number that process $P_{i}$ requests in the current invocation
to enter the CS.
The second one is \emph{Token}, an integer array of size $N$ and \emph{Token}$[i]$ represents the token number selected by process $P_{i}$.
The third shared variable is \emph{Choosing}, a boolean array of size $N$ and \emph{Choosing}$[i]$ is true would indicate that process $P_{i}$ is currently attempting to make a new request to enter the critical section.
The \emph{Session} array and the \emph{Token} array are initialized to zero and the \emph{Choosing} array is initialized to false.
The doorway of the algorithm consists of lines 3-6 and the waiting room is made up of lines 7-10.

When a process leaves the remainder section, it first sets the variable \emph{Choosing}$[i]$ to true to signal to other processes that it is about to make a new request to enter the critical section.
Next, it places the desired session number in \emph{Session}$[i]$.
We assume that all session numbers are positive integers.
It then selects its token number to be one more than the maximum of the token numbers of all other processes and places it in \emph{Token}$[i]$.
Finally, $P_{i}$ sets \emph{Choosing}$[i]$ to false to signal to other processes that it
has completed the doorway section for the new request.

\begin{algorithm}
\caption{Generalized Lamport's Bakery Algorithm}
\label{GLBAlgo}
\begin{algorithmic}[1]
          \REQUIRE ~~ 

              \emph{Session}: \textbf{array}$[1..N]$ of \textbf{integer}, initially all $0$ 

              \emph{Token}: \textbf{array}$[1..N]$ of \textbf{integer}, initially all $0$

              \emph{Choosing}: \textbf{array}$[1..N]$ of \textbf{boolean}, initially all false

            \medskip

          \ENSURE ~~ \\
            \emph{mysession}: \textbf{integer}, initially $0$

        \medskip

       \LOOP

                \STATE \textbf{REMAINDER SECTION}

            \medskip

            \STATE $\emph{Choosing}[i] := \TRUE$

            \STATE $\emph{Session}[i] :=  \emph{mysession}$

            \STATE $\emph{Token}[i] := 1 $ + max of other token numbers

            \STATE $\emph{Choosing}[i] := \FALSE$

            \medskip

            \FOR {$j := 1$ \TO $N$}

                \STATE \textbf{wait until} $((\emph{Choosing}[j] =  \FALSE) \lor (\emph{Session}[j] \in \{0, \emph{mysession}\}))$

                \STATE \textbf{wait until} $(((\emph{Token}[i], i) < (\emph{Token}[j],j)) 
\lor (\emph{Token}[j] = 0) \lor$ \\ \ \ \ \ \ \ \ \ \ \ \ \ \ \ \ \     $(\emph{Session}[j] \in \{0, \emph{mysession}\}))$

            \ENDFOR

            \medskip

            \STATE \textbf{CRITICAL SECTION}

            \medskip

            \STATE $\emph{Token}[i] := 0$

            \STATE $\emph{Session}[i] := 0$

        \ENDLOOP

    \end{algorithmic}
\end{algorithm}

 In the waiting room, at line 8, for each other process $P_{j}$,
process $P_{i}$ checks to see whether \emph{Session}$[j]$ is zero
or same as {\em mysession}. In either case, there is no
problem and so $P_{i}$ can move on. On the other hand,
if process $P_{j}$ is requesting a conflicting session
(i.e., \emph{Session}$[j]  \not \in \{0, $\emph{mysession}$\}$),
$P_{i}$ waits for $P_{j}$ to
finish the doorway and set \emph{Choosing}$[j]$ to be false.
Next, at line 9, process $P_{i}$ waits on each conflicting process $P_{j}$
(i.e., \emph{Session}$[j]  \not \in \{0, $\emph{mysession}$\}$)
until either $P_{j}$ has exited the critical section (i.e., \emph{Token}$[j] = 0$)
or it has a larger token number (i.e., $((\emph{Token}[i], i) < (\emph{Token}[j],j))$).
It is possible that two conflicting processes pick the same token number.
In that case, we use the process identifier to resolve the ties (line 9).
The relation ``less than'' among ordered pairs of integers is defined as in the Lamport's Bakery Algorithm.
That is $(a, b) < (c, d)$ if $a < c$ or if $a = c$ and  $b < d$ (a, b, c and d are all integers).
For simplicity, if \emph{Token}$[i]$ = \emph{Token}$[j]$, we say the process with the smaller process identifier has the smaller token number.
After the loop, process $P_{i}$ enters the CS.
In the exit section, $P_{i}$ resets  \emph{Token}$[i]$ to $0$
and then \emph{Session}$[i]$ to $0$.

It is easy to see that process $P_{i}$
enters the critical section if it has the smallest
token number among all conflicting processes. This
observation immediately implies that the algorithm
has the FCFS property and the mutual exclusion
property. It is trivial to notice that the algorithm
has bounded exit property as the exit section is
made up of two simple write statements. It is also easy
to see that  this algorithm satisfies
the concurrent entry property as no process waits
on another process with the same session number. This algorithm is
also deadlock free as there cannot be a circular
wait among processes because one process must have the smallest token number.
Any algorithm that is
both deadlock free and has FCFS property, also
has starvation freedom property and so we can
conclude that this algorithm has
starvation freedom property. Thus this algorithm
has all the five desired properties P1 through P5.
It is not too difficult to see that this algorithm has
$O(N)$ RMR complexity and is trivial to note that
it is of $O(N)$ shared spaced complexity. 
We summarize the results in the form of the following
theorem. A more formal presentation of the proof is available
in Appendix I.

\begin{thm}
\label{GLBcorrect}
The Generalized Lamport's Bakery Algorithm presented in Fig.~\ref{GLBAlgo} solves
the GME problem by satisfying all the five properties P1 through P5 in linear space (with unbounded registers) and time (RMR under the CC model) using only simple read and 
write operations.
\end{thm}

An earlier paper by Takamura and Igarashi \cite{TI03} also made an attempt to generalize
Lamport's Bakery Algorithm to solve the GME problem.
They presented three different algorithms in
that paper. 
Their first algorithm is fairly simple, but does not satisfy
the starvation freedom property.
Their second and third algorithms, apart from being quite complicated,
do not satisfy the  bounded exit property. All three of their algorithms
do not satisfy the FCFS property  and also the concurrent entry property. 
However, all three of their 
algorithms satisfy a weaker property
known as {\em concurrent occupancy} in the literature 
(see \cite{KM99} and \cite{Hadz01}).
To the best of our knowledge, our algorithm shown
in Fig.~\ref{GLBAlgo} is the simplest and
most elegant generalization of Lamport's Bakery
Algorithm for the GME problem.

\section{A Flaw in the Literature}
\label{flaw}

   We now look at attempts to solve the GME problem
using only bounded shared registers and simple read and write
operations. Two prominent algorithms in this regard are
that of Hadzilacos \cite{Hadz01} and that of 
Jayanti et al. \cite{JPT03}. Both of the algorithms
can be viewed as a modular composition of an ``FCFS Algorithm'' and an ``ME Algorithm''. 

Assume that processes $P_{i}$ and $P_{j}$ are two active
processes requesting different sessions.
It is possible that 
neither $P_{i}$ doorway precedes $P_{j}$, nor $P_{j}$ doorway
precedes $P_{i}$. This will be the case if one process enters its
doorway after another process has entered, but not completed, its doorway.
In that case, we would call the two processes as {\em doorway concurrent}
processes. It is to be noted that the FCFS property does not dictate
as to which one of two {\em doorway concurrent} processes should enter
the CS first.
In fact, both processes can get out of the ``FCFS
algorithm'' and can potentially enter the CS
simultaneously. In order to prevent this from
happening, we need the ``ME Algorithm''. 

Not only that both of the algorithms use modular
composition technique, they both use the same ``ME Algorithm'' 
viz., a slightly modified version of a classical mutual exclusion algorithm independently discovered by Burns \cite{Burns81} 
and Lamport \cite{Lamp86}.  Hadzilacos's algorithm is
of space complexity $\Theta(N^{2})$ and he claimed that his
algorithm is of $O(N)$ RMR complexity in the CC Model.
Hadzilacos posed it as an open problem to devise 
an algorithm to solve the GME problem in linear time
and space using only bounded shared variables.
Jayanti et al. came up with a clever modification to the
``FCFS algorithm'' of Hadzilacos and reduced 
the space complexity to $\Theta(N)$ and tacitly
inherited the claim that algorithm is of 
$O(N)$ RMR complexity.

We prove that the ``Burns-Lamport ME Algorithm'' is of RMR complexity $\Omega(N^2)$ in the CC Model thus invalidating their claims.
To give a intuition of the analysis, consider $N$ doorway concurrent processes request the CS, we can construct a scenario in which each process $P_{i}$ will be blocked by any other process $P_{j}$ exactly $j$ times (for any $j<i$) before $P_{i}$ enters the CS.
Therefore, process $P_{N}$ will be blocked exactly $\frac{N(N-1)}{2}$ times before it enters the CS.
As each block involves a constant number of RMR, the worst case RMR complexity is $\Omega(N^2)$.
The formal proof the complexity analysis of the ``Burns-Lamport ME Algorithm'' appears in the Appendix II.
Hence the problem of developing a linear time (RMR) and
linear (shared) space algorithm that uses only bounded
shared variables for the GME
problem, originally posed by Hadzilacos, is
still open.

One might get the impression that we can immediately fix
the problem, by plugging in some other mutual exclusion
algorithm that has $O(N)$ RMR Complexity 
in place of Burns-Lamport  Mutual Exclusion Algorithm.
Unfortunately, the situation is not that simple as
we have to adapt the ME algorithm so that it provides
concurrent entry for application in the development of GME
algorithm using the modular composition technique.
The Lamport-Burns algorithm is easy to adapt by
simply adding an extra condition to check whether 
the session number of the other process is same as
the session number of this process
in all wait-until loops.
On the other hand, the mutual exclusion algorithms  Taubenfeld \cite{Tau04}
which has $O(N)$ RMR Complexity or that of
Yang and Anderson \cite{YA95} which has $O(\log N)$ RMR
Complexity, are not easily adaptable to provide
concurrent entry. Simply adding an extra condition to
check the session number in all wait-until loops in these
algorithms does not work as they have more intricate structure.
We are unable to find a mutual exclusion algorithm 
of $O(N)$ RMR complexity in the literature, that is easily 
adaptable to provide concurrent entry.
So, the problem of developing a linear time and linear space
GME algorithm that uses only bounded registers is indeed a
non-trivial problem. In the next section, we develop
such an algorithm.

\section{Black and White Bakery GME Algorithm}
\label{bwgme}

In 2004, Taubenfeld \cite{Tau04} came up with an elegant algorithm
called {\bf Black and White Bakery Algorithm} that
solves the classical mutual exclusion problem with
only bounded shared registers. In this section, we generalize
the ideas developed in that paper and solve the GME 
problem using bounded registers in linear time and
space. 

The algorithm uses a multi writer multi reader shared
bit variable called {\em GlobalColor} (see Figure~\ref{header-black-white-GME})
which can only be
black or white. All other shared variables used
in the algorithm can only be written by one process
even though they can be read by multiple processes.
Each process uses a shared variable called {\em Token} that has
three components viz., \emph{session}, \emph{color} and \emph{number}.
We assume that processes can read or write into this \emph{Token} variable
atomically even though it has three components. This is not
an unreasonable assumption as this can be implemented
without the aid of any higher level synchronization primitives
by encoding three integers into a single integer using
simple techniques (which we are not elaborating further here).
Finally each process also has a boolean shared variable called
{\em Choosing}.  Also, unlike the \emph{GlobalColor} variable, the
token color of a process can be black or white or a special
value denoted by $\perp$ which indicates that the process
has not yet set its token color.

The algorithm is depicted in Figure~\ref{BWGMEAlgo}
and we will refer to it as the Black and White Bakery GME
(BWBGME) Algorithm.
The doorway of the algorithm is made up of lines 3-15 and the waiting room section consists of lines 16-23.
When a process leaves the remainder
section, it picks a session number and then updates
its {\em Token} variable to reflect it. It then sets
its {\em Choosing} variable to be true to indicate to
other processes that it has initiated the task of picking
a token. The key difference this time is that the
token is a colored token i.e., the token
has both a color and a number. The token color is set to
be same as the current value of the shared variable {\em
GlobalColor}. The token number is set to
be one more than the maximum of token numbers of
conflicting processes with the same color and is set to
be $1$ in case there are no conflicting processes with
the same color. The process then updates its {\em Token}
variable to reflect the chosen color and number (line 14). It then
sets its {\em Choosing} variable to be false to indicate
to other processes that it is done with the task of
picking a colored token (line 15).

%==============================================
%   Header for Black-White Bakery GME Algorithm
%==============================================
\begin{algorithm}[b]
    \caption{Header for Black and White Bakery GME Algorithm in Fig.~\ref{BWGMEAlgo}}
    \begin{algorithmic}[1]
    \label{header-black-white-GME}

       \REQUIRE ~~ 

        \emph{GlobalColor}: a \textbf{bit} of type \{\textbf{black}, \textbf{white}\}, initialized arbitrarily

        \emph{Token}: \textbf{array}$[1..N]$ of 
        (\emph{session}: \textbf{integer}, \emph{color} $\in \{\textbf{black}, \textbf{white},\perp\}$, \emph{number}: \textbf{integer}), initially all $( 0,$ $\perp,$ $0)$ 

          \emph{Choosing}: \textbf{array}$[1..N]$ of \textbf{boolean}, initially all false

        \medskip

        \ENSURE ~~ 

            \emph{mysession}: \textbf{integer}, initially $0$

            \emph{mycolor}: a \textbf{bit} of type \{\textbf{black}, \textbf{white}\}, initialized arbitrarily

            \emph{mynumber}: \textbf{integer}, initially $0$

            \emph{other}:  (\emph{session}: \textbf{integer}, \emph{color} $\in \{\textbf{black}, \textbf{white},\perp\}$, \emph{number}: \textbf{integer}), initially  $( 0,$ $\perp,$ $0)$ 

    \end{algorithmic}
\end{algorithm}

%==============================================
%  Black-White Bakery GME Algorithm
%==============================================
\begin{algorithm}
    \caption{Black and White Bakery GME Algorithm}
    \begin{algorithmic}[1]
    \label{BWGMEAlgo}
        \LOOP

        \STATE \textbf{REMAINDER SECTION}

        \medskip

        \STATE $\emph{Token}[i] := (\emph{mysession} , \perp, 0)$

        \STATE $\emph{Choosing}[i] :=  \TRUE$

        \STATE $\emph{mycolor} :=  \emph{GlobalColor}$

        \STATE $\emph{mynumber} := 0$

        \FOR {$j := 1$ \TO $N$}

            \STATE $\emph{other} :=  \emph{Token}[j]$

            \IF{$((\emph{other.color} = \emph{mycolor} )\land (\emph{other.session} \not\in \{ 0, \emph{mysession}\}))$}

                \STATE  $\emph{mynumber} := $ max$(\emph{other.number}, \emph{mynumber})$
            \ENDIF

        \ENDFOR

        \STATE $\emph{mynumber} :=\emph{mynumber} + 1 $

        \STATE $\emph{Token}[i]:=  (\emph{mysession}, \emph{mycolor}, \emph{mynumber})$

        \STATE $\emph{Choosing}[i] :=  \FALSE$

        \medskip

        \FOR {$j := 1$ \TO $N$}

          \STATE \textbf{wait until}
                $((\emph{Choosing}[j] =  \FALSE) \lor
                (\emph{Token}[j].\emph{session} =  \emph{mysession}))$

          \IF {$\emph{Token}[j].\emph{color}  =  \emph{mycolor}$}

              \STATE \textbf{wait until}
                    $(( (\emph{mynumber}, i) < (\emph{Token}[j].\emph{number}, j)) \lor
                    (\emph{Token}[j].\emph{color} \neq  \emph{mycolor}) \lor $\\
                    \ \ \ \ \ \ \ \ \ \ \ \ \ \ \ \ \
                    $(\emph{Token}[j].\emph{session} \in\{0, \emph{mysession}\}))$

            \ELSE
        \STATE \textbf{wait until}
                    $((\emph{GlobalColor} \neq \emph{mycolor}) \lor
                    (\emph{Token}[j].\emph{color} =  \emph{mycolor}) \lor $\\
                    \ \ \ \ \ \ \ \ \ \ \ \ \ \ \ \ \ 
                    $(\emph{Token}[j].\emph{session} \in\{0, \emph{mysession}\}))$

            \ENDIF

        \ENDFOR

        \medskip

        \STATE \textbf{CRITICAL SECTION}

        \medskip

        \IF {$\emph{mynumber} \neq 1$}
            \IF {$($\textbf{not} OPPOSITECOLOR$(\emph{mycolor}))$}
                \IF{$\emph{mycolor}  =  \textbf{black}$}
                    \STATE $\emph{GlobalColor} :=  \textbf{white}$
                \ELSE
                    \STATE $\emph{GlobalColor} :=  \textbf{black}$
                \ENDIF
            \ENDIF
        \ENDIF

        \STATE $\emph{Token}[i]:= ( 0,\perp,0)$

    \ENDLOOP

\end{algorithmic}
\end{algorithm}

%==============================================
%   Method: OPPOSITECOLOR
%==============================================
\begin{algorithm}
    \caption{Method OPPOSITECOLOR$(\emph{color})$ for Black and White Bakery GME Algorithm in Fig.~\ref{BWGMEAlgo}}
    \begin{algorithmic}[1]
    \label{method-color-black-white-GME}

        \FOR {$j := 1$ \TO $N$ }

        \STATE $\emph{other} := \emph{Token}[j]$

            \IF {$((\emph{other}.\emph{session} \neq 0) \land (\emph{other}.\emph{color} = \overline{\emph{color}}))$}

                \RETURN \TRUE

            \ENDIF

        \ENDFOR

        \RETURN \FALSE

    \end{algorithmic}
\end{algorithm}

In the waiting room, for each other process $P_{j}$, process $P_{i}$ checks to see
if it is an active conflicting process. If it is not,
then there is no problem and $P_{i}$ does not wait
on $P_{j}$. If it is, then process $P_{i}$ waits until
$P_{j}$ has completed selecting its token color and
number, if it has initiated the task (line 17). Process $P_{i}$ then
checks whether it has priority over $P_{j}$ (lines 18-22).
If so, it does not wait on process $P_{j}$ and otherwise
it waits on $P_{j}$.
The priority order between conflicting processes 
is defined as follows. If two conflicting processes have different token colors,
the process whose token color is different from the {\em GlobalColor}
gets the priority. On the other hand, if two conflicting processes have the
same token color, the process with the smaller token
number gets the priority.  If two conflicting processes have the same token color 
and number, then the process with the smaller identifier 
gets priority. Process $P_{i}$ enters the CS if it has priority over all other conflicting processes.

At the time of exiting, process $P_{i}$ checks whether its token
number is $1$ (line 25). If so, it just resets its {\em Token} variable
to the initial value and exits. If not, it checks whether
there is an active process (not necessarily a conflicting process) with the {\em opposite token color} (see Figure~\ref{method-color-black-white-GME}).
The {\em opposite token color}, denoted by  $\overline{color}$, is defined to be \emph{black} 
if $\emph{color}= $ white, and vice versa.
If so, it just resets its {\em Token} variable to the
initial value and exits. If not, then it updates the {\em GlobalColor} to the opposite
of its token color and then resets its {\em Token} variable to
the initial value.  In particular, note that token color is reset to $\perp$.  
This is important to ensure that other processes do not erroneously
use a process's old token color while determining the priority.

The generalization of the Black and White Bakery Algorithm to solve the
GME problem is quite tricky. While the formal proof of correctness can be found
in the Appendix III of this paper, we provide some main insights into our generalization
here.
In the original Black and White Bakery Algorithm, when a process is attempting to
enter the CS, it selects its token color as the current global color
and its token number to be one more than the maximum of token numbers of processes with the
same color. When it exits the CS, it simply updates the \emph{GlobalColor} to be the opposite of its own token color.

A naive generalization of Black and White Bakery Algorithm would simply add an additional check 
to see whether the other process is a conflicting process in all busy-wait loops. 
This naive generalization, nevertheless, can not ensure the correctness in the case of
group mutual exclusion. In the GME, as processes with the same session and 
different token colors can be in the CS at the same time, if a process leaving 
the CS simply updates the \emph{GlobalColor} as before, it may erroneously allow conflicting processes to stay in the CS simultaneously.

Consider the following scenario.  Initially, the \emph{GlobalColor} is white.
Processes $P_{i}$ and $P_{j}$ request the session $S$ and get the token color of white.
Then, $P_{i}$ and $P_{j}$ enter the CS concurrently because there is no conflict.
When $P_{i}$ is exiting, it sets the \emph{GlobalColor} to black. 
Next, a process $P_{k}$ requests the same session $S$ and gets its token color of black.
It is easy to see that $P_{k}$ enters the CS by the concurrent entry property.
After that, a process $P_{l}$ starts to request a conflicting session $S'$ and gets the token color of black.
In the waiting room, $P_{l}$ waits for $P_{j}$ since they have different token color and the \emph{GlobalColor} is black.
However, if we let $P_{k}$ exit the CS and then simply set the \emph{GlobalColor} to white, $P_{l}$ will stop waiting for $P_{j}$ as it sees the \emph{GlobalColor} is different than its token color.
Hence, two conflicting processes $P_{j}$ and $P_{l}$ will be in the CS simultaneously, thus violating the mutual 
exclusion property.

      To get a correct generalization, we observed a key invariant of the 
original Black and White Bakery Algorithm.  

% {\bf Key Invariant:} 
%     After a process gets its token color from the  \emph{GlobalColor}, the  \emph{GlobalColor} can not
% be flipped twice, before the process finishes the critical section and gets out of the exit section.

\begin{inv}
After a process $P_{i}$ gets its token color from the GlobalColor, the GlobalColor can not
be flipped twice, before the process $P_{i}$ finishes the critical section and gets out of the exit section.
\end{inv}

In the original Black and White Bakery Algorithm, supposing a process $P_{a}$ gets a token color of
black and after some time the \emph{GlobalColor} gets changed to white by some process $P_{b}$. Now,
in order for some other process $P_{c}$ to change it again to black, $P_{c}$ must be a white
process. However, $P_{c}$ can change it only while it exits. In order to exit, $P_{c}$ must
first enter the critical section. As \emph{GlobalColor} is currently white 
and $P_{a}$ has a color of black, $P_{a}$ will have the higher priority over $P_{c}$ to
enter the critical section. So, the \emph{GlobalColor} cannot be changed again until $P_{a}$
gets out completely.

The fundamental idea in generalizing the Black and White Bakery Algorithm is to ensure that this invariant is maintained. It is not difficult to see that doing so, solves the mutual exclusion violation illustrated in the previous scenario.
Although processes $P_{j}$ and $P_{k}$ with different token color stay in the CS at the same time, the
\emph{GlobalColor} will not be updated when $P_{k}$ executes the exit section (as otherwise the \emph{GlobalColor} is flipped twice since $P_{j}$ got its token color and before it exits) and therefore,
$P_{l}$ will still wait for $P_{j}$ until it exits.

In order to maintain this invariant in the generalization, when a process finished the CS, 
we let it check whether there is another active process with the opposite token color.
If there exists such a process, then the exiting process does not update the \emph{GlobalColor}.
Otherwise, the process updates the \emph{GlobalColor} to be the opposite color of its own token color.
Also, while a process is checking this condition, it may unintentionally access the old 
token color of another process even though that process already finished the 
previous invocation.  To prevent this from happening, we let processes reset their 
token color to empty $(\bot)$ at the end of the exit section.

     However, only the new \emph{GlobalColor} updating mechanism is not enough to keep this invariant.
A process in the doorway may read the (opposite) \emph{GlobalColor} and then stop before writing it 
to its token color.  Hence, another process in the exit section would have no clue as to which color this process will get
(it will see a color of $\bot$ for this process). It could erroneously think there is no active processes with
opposite token color and flip the \emph{GlobalColor} (while in fact there is an active process with the
opposite color).  We can devise an intricate sequence of execution to show that the key invariant 
does not hold if there is such a process that stopped just before writing down its token color in the doorway.

In order to handle this, we use different scheme for token number picking (lines 7-12) and 
add another condition to check before updating the \emph{GlobalColor} (line 25).
When a process is picking a token number, it determines the maximum of token numbers of conflicting processes with the same token color and then increments it by 1.
If there is no such process, it selects its token number to be 1.
At the time of exiting, a process does not even attempt to update the \emph{GlobalColor} if its token
number is 1. On the other hand, if its token number is 2 or more, then it attempts to update the
\emph{GlobalColor} (it actually does if there is no active process with the opposite token color).

These changes are necessary to solve the issue of overlooking a process with the opposite token color having a ``$\perp$'' value.
Suppose that a hanging process $P_{h}$ has read a value of black for \emph{GlobalColor} and has
not yet written into its token color variable. Suppose that some other process changed the
\emph{GobalColor} from black to white after some time. Later, if a process $P_{u}$ starts to execute the algorithm and  attempts to
change the \emph{GlobalColor} again, then $P_{u}$ must be a white-colored process with a token number
greater or equal to 2.
In order for process $P_{u}$ to get a token number of $2$ or more, there must exist an active conflicting white-colored process $P_{v}$ with a smaller token number.
Clearly, one of $P_{u}$ and $P_{v}$ must be in conflict 
with the hanging process $P_{h}$ because either $P_{u}$ or $P_{v}$ has a different session with $P_{h}$.  Therefore, at least one of $P_{u}$ or $P_{v}$ would 
have waited (in line 17) for $P_{h}$ to finish the token selection before it entered the CS because $P_{h}$ is hanging at the doorway when $P_{u}$ and $P_{v}$ enter the waiting room.
Whichever be the case, as $P_{u}$ can enter the CS only after $P_{v}$ 
has left the CS, we can conclude that the ``hanging process'' $P_{h}$ 
has really written down the read \emph{GlobalColor} to its token variable by the time $P_{u}$ is checking $P_{h}$'s token color in the exit section. 
This shows that the scenario that we mentioned before cannot possibly occur anymore.
A formal proof that our generalization maintains this key invariant is available in the Appendix III 
(see Lemma \ref{lemma3}). 

The token numbers used in our algorithm can not grow beyond $N+1$.
To give an intuition of the bound, consider the following example.
Initially, the \emph{GlobalColor} is white.
Suppose all processes request different sessions and enter the doorway one by one.
The first process will get a token color of white and a token number of $1$.
Each other process will get the same token color of white and progressively
increasing token numbers (by incrementing the token number of the previous process).  
Obviously, the last process gets the token number of $N$.  The process with the token number of 
$1$ will exit without updating the \emph{GlobalColor}.  If that process requests 
a conflicting session again,
it will get the token number of $N+1$.  When a process with the token 
number of $2$ or more is exiting, it will set the \emph{GlobalColor} to black as 
no process has the opposite token color.  So, later processes will have the 
token color of black and will pick the token number starting from $1$.
A formal proof of this fact is available in Appendix III (see Theorem \ref{bwbound}).
Note that although the session variables are unbounded,
it is due to the application
and is not an artifact of the algorithm.

It is trivial to observe that the algorithm uses $O(N)$ shared space.
A complete proof of the $O(N)$ RMR complexity is available in Appendix III (see Theorem \ref{bwcomp}).
We summarize the result stating the properties of our
algorithm in the following theorem. 

\begin{thm}

   The BWBGME Algorithm presented
in Fig.~\ref{BWGMEAlgo} solves the GME problem by satisfying
all the five properties P1 through P5 in linear space (with
bounded registers) and time (RMR under the CC model) using
only simple read and write operations.

\end{thm}

Even though the token numbers in our algorithm are bounded, the bound is $N+1$
which is not a constant.  We leave the development of
a linear time and linear space GME algorithm satisfying all five properties,
which uses only simple read and
write instructions and whose shared registers are bounded by a constant,
as an open problem. 

\newpage

\section*{Acknowledgments}
\label{ack}

    The authors would like to thank an anonymous referee who found a
subtle error in an earlier version of our Black and White Bakery GME algorithm.

\newpage

%=====================================================
%   Appendix I
%=====================================================

\section*{Appendix I: Proof of Correctness of the Generalized Lamport's Bakery Algorithm}

In this appendix, we present a formal proof of Theorem~\ref{GLBcorrect} stated in Section~\ref{simple}.
More specifically, we show that the Generalized Lamport's Bakery Algorithm
algorithm satisfies the properties P1 through P5 and runs in linear time and space.

%=====================================================
%   Theorem: FCFS
%=====================================================
\begin{thm}
\label{fcfs}
The GLB algorithm has the FCFS property.
\end{thm}

\begin{proof}
Assume two conflicting processes $P_{i}$ and $P_{j}$ request the CS.
Also assume that $P_{i}$ finished the doorway before $P_{j}$ starts the doorway.
So, when process $P_{j}$ is computing its token number, it will read the token number of $P_{i}$ and then select a larger token number.
As $P_{j}$ has a larger token number, when $P_{j}$ checks with $P_{i}$ at line 9, it waits for $P_{i}$ to exit the CS.
\qed
\end{proof}

%=====================================================
%   Theorem: Mutual Exclusion
%=====================================================
\begin{thm}
\label{me}
The GLB algorithm has the Mutual Exclusion property.
\end{thm}

\begin{proof}
Suppose two conflicting processes $P_{i}$ and $P_{j}$ are in the CS at the same time with different sessions.
By Theorem \ref{fcfs}, $P_{i}$ and $P_{j}$ must be doorway-concurrent.
Without loss of generality, we assume process $P_{i}$ enters the CS first.
Since $P_{i}$ and $P_{j}$ are doorway concurrent, $P_{j}$ has at least finished executing line 3 when $P_{i}$ checks on $P_{j}$ at line 8.

If process $P_{j}$ has also finished updating \emph{Session}$[j]$, but not the whole doorway by that time, $P_{i}$ will wait for $P_{j}$ to finish the doorway, as they are requesting conflicting sessions. 
So, when process $P_{i}$ passes line 8, $P_{j}$ will have a valid token number.
Since process $P_{i}$ enters the CS first, it must have a smaller token number than $P_{j}$.
As a consequence, at line 9, process $P_{j}$ will wait for $P_{i}$ because it finds that $P_{i}$ has a smaller token number.
Therefore, $P_{j}$ can not enter the CS until $P_{i}$ finishes the CS and resets its token number, which is contradicting with our assumption.

On the other hand, if process $P_{j}$ has not finished updating \emph{Session}$[j]$ by that time, then process $P_{i}$ will not wait on $P_{j}$ in line 8 as \emph{Session}$[j] = 0$.
However, if this be the case, $P_{i}$ has already selected its token number whereas $P_{j}$ has not yet begun selecting its token number. 
So, process $P_{j}$'s token number is guaranteed to be larger than that of $P_{i}$ and so $P_{j}$ will wait on $P_{i}$ at line 9.
Therefore, process $P_{j}$ can not enter the CS until $P_{i}$ finishes the CS and resets its token number, which is contradicting with our assumption.
\qed
\end{proof}

%=====================================================
%   Theorem: Bounded Exit
%=====================================================
\begin{thm}
\label{be}
The GLB algorithm has the Bounded Exit property.
\end{thm}

\begin{proof}
Since the exit section consists of two simple write instructions, a process that enters the exit section will trivially finish it within a  bounded number of its own steps.
\qed
\end{proof}

%=====================================================
%   Theorem: Deadlock Freedom
%=====================================================

We now define another property called {\em Deadlock Freedom}.
Deadlock freedom simply means deadlocks cannot occur in the system.
Informally, {\em Deadlock} occurs in the system when one or more processes are ``trying to enter'' the CS, but no process ever does so.
Lamport (see page 329 in \cite{Lamp86}) has shown that under the assumptions that no process stays in the CS forever (which we are assuming), there are only finitely many processes (which we are assuming) and the algorithm satisfies the ``Bounded Exit property'' (which we have just now shown is the case for our algorithm), the deadlock freedom property can be formally stated as follows.

\begin{description}
\item[P6] {\bf Deadlock Freedom} If one or more processes are forever trying to
enter their critical section with no success for any of them, then there exists
a process that enters the critical section infinitely often.
\end{description}

\begin{thm}
\label{df}
The GLB algorithm has the Deadlock Freedom property.
\end{thm}

\begin{proof}
Suppose the algorithm does not satisfy the deadlock freedom property.
Then there is an execution of the algorithm in which a nonempty set $S$ of processes enter the entry section but none of them enters the CS, and no process enters the CS infinitely often.
We observe that no process can wait at line 8 forever since every process $P_{j}$ that
requests the CS will finish the doorway eventually and set $\emph{Choosing}[j]$ to false and process $P_{j}$ enters the CS only finitely many times.
Therefore, processes in the set $S$ must wait on line 9 forever.
There exists a process $P_{i}$ in set $S$ that has the smallest token number among all processes in $S$ and that $P_{i}$ will pass all other processes at line 9 and enter the CS, which is a contradiction with our assumption.
\qed
\end{proof}

%=====================================================
%   Theorem: Starvation Freedom
%=====================================================

It is easy to observe that the {\em Starvation Freedom} property immediately implies the {\em Deadlock Freedom} property.
However, the converse is not necessarily true.
Deadlock freedom means that the entire system of processes will always continue to make progress.
However, it does not preclude the possibility of individual processes not making progress (i.e., waiting forever in the entry section).

\begin{thm}
\label{sf}
The GLB algorithm has the Starvation Freedom property.
\end{thm}

\begin{proof}
Lamport proved (see page 330 in \cite{Lamp86}) that if an algorithm satisfies the deadlock freedom property and the FCFS property, then it necessarily satisfies the starvation freedom property.
Hence, combining Theorem~\ref{fcfs} and Theorem~\ref{df}, we conclude that our
algorithm has the starvation freedom property.
\qed
\end{proof}

%=====================================================
%   Theorem: Concurrent Entry
%=====================================================
\begin{thm}
\label{ce}
The GLB algorithm has the Concurrent Entry property.
\end{thm}

\begin{proof}
Assume that process $P_{i}$ requests a session, and no other process requests a different session.
This means, for every other process $P_{j}$, it holds that \emph{Session}$[j] \in \{0, $\emph{mysession}$\}$ from the view of $P_{i}$.
Since $P_{i}$ always checks whether \emph{Session}$[j] \in \{0, $\emph{mysession}$\}$ in both the busy-wait loops in line 8 and line 9, it cannot wait on either lines.
So, it enters the CS within a bounded number of its own steps, thus satisfying the concurrent entry property.
\qed
\end{proof}

%=====================================================
%   Theorem: Remote Memory Reference Complexity 
%=====================================================
\begin{thm}
\label{rmrc}
The GLB algorithm has $O(N)$ Remote Memory Reference Complexity. 
\end{thm}

\begin{proof}
We analyze the RMR complexity of the algorithm.
Recall that in the CC model, all shared variables are stored in a global memory module
and processes migrate them to their local cache to access them.
In the waiting room, there are only two loops viz., the busy-wait loops in line 8 and line 9.
In line 8, when process $P_{i}$ is busy waiting for a process $P_{j}$, if \emph{Choosing}$[j]$ changes to false, then $P_{i}$ will immediately terminate the wait.
It is possible that (before $P_{i}$ observes the changed value of \emph{Choosing}$[j]$)
process $P_{j}$ sets the \emph{Choosing}$[j]$ to true again and requests another conflicting
session. In that case, since $P_{i}$ doorway proceeds $P_{j}$, process $P_{j}$ will get a  larger token number than $P_{i}$.
So, $P_{j}$ can not enter the CS to finish that invocation and therefore can not change the \emph{Session} again until $P_{i}$ finishes the CS.
Thus, line 8 can only involve a maximum of five RMR (three for \emph{Choosing}$[j]$ and two for \emph{Session}$[j]$).
Similarly in line 9, when process $P_{i}$ is busy-waiting on \emph{Token}$[j]$, if \emph{Token}$[j]$
changes, the new value of it will be either zero or a larger token number and
in either case, $P_{i}$ will terminate the wait. Also, any change in
\emph{Session}$[j]$ will also entail a change in \emph{Token}$[j]$ and thus
terminate the wait.
Hence, line 9 involves a maximum of five RMR (one for \emph{Token}$[i]$, two for \emph{Token}$[j]$, two for \emph{Session}$[j]$).
There are only constant number of RMR in line 8 and line 9.
As these two lines are enclosed within a for loop that can run a maximum of $N$ times, it follows that the entire waiting room section is of $O(N)$ RMR complexity.
Note that the doorway made up of line 3 through line 6 involves only constant number of RMR, except for the implicit loop in line 5.
The implicit loop in line 5 has $O(N)$ RMR complexity as it involves inspecting the token numbers of all other processes.
Finally, it is easy to see that the exit section consisting of lines 12-13 involves exactly two RMR.
Hence, the overall RMR complexity of this algorithm in the CC model is $O(N)$.
\end{proof}

  It is trivial to observe that the Generalized Lamport's Algorithm uses only $O(N)$ shared
space. Hence, the correctness of Theorem~\ref{GLBcorrect} is now established.

%=====================================================
%   Appendix II
%=====================================================

\section*{Appendix II: Analysis of Burns-Lamport Mutual Exclusion Algorithm}

In this appendix, we show that the ``ME algorithm'' used by Hadzilacos \cite{Hadz01} and Jayanti et al. \cite{JPT03} are of RMR complexity $\Omega(N^{2})$.
This elegant ``ME algorithm'' is independently discovered by Burns \cite{Burns81} and Lamport \cite{Lamp86}.
This algorithm is depicted in Figure~\ref{BLalgo}.
In this algorithm \emph{Competing} is a shared array of size $N$.
Each element of the array is a boolean variable, initialized to false.
The code given in Figure~\ref{BLalgo} is for process $P_{i}$.
The variable $j$ is a private variable.

We now briefly describe the algorithm.
Every process $P_{i}$ owns a single bit \emph{Competing}$[i]$.
Only $P_{i}$ can write into \emph{Competing}$[i]$; Other processes can read it.
Before entering the CS, $P_{i}$ sets its bit to true and checks all lower numbered processes (lines 1-3).
If any of them, say process $P_{j}$, is found to have set its bit to true, then process $P_{i}$ resets its bit to false, allowing the smaller-numbered process to make progress.
It then waits for process $P_{j}$'s bit to become false and then restarts the competition to enter the CS by going to line 1.
Having checked all lower-numbered processes, process $P_{i}$ then checks the higher-numbered ones and waits for each of them to set its bit to false.
Now, however, while $P_{i}$ is waiting it does not set its bit to false.
After that process $P_{i}$ enters the CS.
It turns out that this simple algorithm guarantees mutual exclusion.
We refer the reader to \cite{Lamp86} (see also \cite{Burns81}) for a proof of correctness.

\begin{algorithm}[t]
\caption{Burns-Lamport ME algorithm}
\label{BLalgo}
\begin{algorithmic}[1]
\STATE \textbf{L}: \emph{Competing}$[i] := $ \TRUE
\FOR{$j := 1$ \TO $i-1$}
	\IF{\emph{Competing}$[j]$}
	\STATE \emph{Competing}$[i] := $ \FALSE
	\STATE \textbf{wait until} \textbf{not} \emph{Competing}$[j]$
\STATE goto \textbf{L}
\ENDIF
\ENDFOR
	\FOR{$j := i+1$ \TO $N$}
	\STATE \textbf{wait until} \textbf{not} \emph{Competing}$[j]$
\ENDFOR
\end{algorithmic}
\end{algorithm}

To analyze the RMR complexity under the CC Model, consider the following sequence of events

\begin{enumerate}
	\item Process $P_{N}$ sets its bit to true.

	\item Process $P_{(N-1)}$ sets its bit to true.

	\item Process $P_{N}$ checks all lower-numbered processes and finds that $P_{(N-1)}$'s bit is set.
	So, $P_{N}$ sets its bit to false and waits for \emph{Competing}$[N-1]$ to become false.

	\item Process $P_{(N-2)}$ sets its bit to true.

	\item Process $P_{(N-1)}$ checks all lower-numbered processes and finds that process $P_{(N-2)}$'s bit is set.
	So, process $P_{(N-1)}$ sets its bit to false and waits for \emph{Competing}$[N-2]$ to become false.

	\item Process $P_{N}$ now finds \emph{Competing}$[N-1]$ to be false and so restarts the competition by setting its bit to true.

	\item Process $P_{N}$ checks all lower-numbered processes and finds that process $P_{(N-2)}$'s bit is set.
	So, $P_{N}$ sets its bit to false and waits for \emph{Competing}$[N-2]$ to become false.

	\item Process $P_{(N-3)}$ sets its bit to true.

	\item Process $P_{(N-2)}$ checks all lower-numbered processes and finds that process $P_{(N-3)}$'s bit is set.
	So, process $P_{(N-2)}$ sets its bit to false and waits for \emph{Competing}$[N-3]$ to become false.

	\item Process $P_{N}$ now finds \emph{Competing}$[N-2]$ to be false and so restarts the competition by setting its bit to true.

	\item Process $P_{N}$ checks all lower-numbered processes and finds that process $P_{(N-3)}$'s bit is set.
	So, $P_{N}$ sets its bit to false and waits for \emph{Competing}$[N-3]$ to become false.

	\item .

	\item .

	\item .

	\item Process $P_{N}$ checks all lower-numbered processes and finds that process $P_{1}$'s bit is set.
	So, $P_{N}$ sets its bit to false and waits for \emph{Competing}$[1]$ to become false.

	\item Process $P_{1}$ checks all higher-numbered processes and finds that all the bits are false and enters CS.

	\item Process $P_{1}$ exits the CS and sets its bit to false.
\end{enumerate}

Note that during the above sequence of events, process $P_{N}$ got blocked by each one of the lower numbered processes once.
At the end of the above sequence, process $P_{1}$'s request is satisfied.
Also, at the end of the above sequence, the \emph{Competing} bit of $P_{2}$ through $P_{N}$ are all false.
Now, we can create similar sequence of events, but this time with only processes $P_{2}$ through $P_{N}$ participating.
We can recursively create similar sequence of events again and again until finally we have only process $P_{N}$ participating.

The net effect is that in this worst case scenario, process $P_{N}$ got blocked by process $P_{(N-1)}$ a total of $(N-1)$ times, by process $P_{(N-2)}$ a total of $(N-2)$ times and so on.
So, the total number of times, process $P_{N}$ gets blocked by some other processes is

\[\sum_{k=1}^{N-1} k = N(N-1)/2 = \Omega(N^2) \]

In the CC model, at least one remote memory reference is involved each time a process gets blocked and hence the worst case RMR complexity of the algorithm in Figure~\ref{BLalgo} is $\Omega(N^2)$ in the CC Model.
It was erroneously claimed in \cite{Hadz01} and tacitly inherited in \cite{JPT03} that this algorithm is of $O(N)$ RMR complexity.
That claim is clearly wrong as illustrated above.

%======================================================
%   Appendix III
%======================================================

\section *{Appendix III: Proof of Correctness of the Black and White Bakery GME Algorithm}

In this appendix, we present a complete proof of the correctness and complexity analysis of Black and White Bakery GME Algorithm.
In particular, we show that our algorithm satisfies the properties P1 through P5, uses only bounded registers and has linear time and space complexity.

We use the notation $P_{i}{@}{x}$ to denote that process $P_{i}$ is executing line $x$ of the algorithm.
The notation $P_{i}{@}x \rightarrow P_{j}{@}y$ means process $P_{i}$ executes line $x$ before process $P_{j}$ executes line $y$.
We use $\overline{color}$ to denote the \emph{opposite} color (i.e, $\overline{\emph{color}} = $ black if $\emph{color} =$ white, and vice versa).
$\overline{color}$ is not defined if $\emph{color} = \perp$.
It is easy to see that \emph{GlobalColor} always has a value of either black or white.
Consider the execution of the algorithm on the global time line.
We use $I_{i}$ to denote the time interval that the \emph{GlobalColor} remains as some color after it has been flipped $(i-1)$ times from the beginning.
Once the \emph{GlobalColor} is changed to the opposite color at the end of $I_{i}$, the \emph{GlobalColor} remains as that color during the interval $I_{i+1}$.
We use $t_{i}$ to represent the time point at which the \emph{GlobalColor} is flipped to the opposite color at the end of $I_{i}$.
Note that, $t_{i}$ will be the starting point of the time interval $I_{i+1}$.
Process $P_{t_{i}}$ denotes the process that flips the \emph{GlobalColor} at $t_{i}$.

Here is an example to show a possible execution of the algorithm.
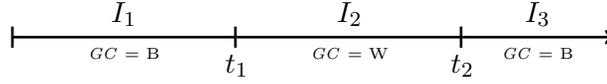
\begin{center}
\begin{tikzpicture}[line width=1pt][every node/.style={midway}]
\draw[|-|] (0,0) --  node[anchor=south] {$I_{1}$} node[anchor=north] {\tiny{\emph{GC} $=$ B}} (3,0)  node[anchor=north,outer sep=3pt] {$t_{1}$} ;
\draw[-|] (3,0) --  node[anchor=south] {$I_{2}$} node[anchor=north] {\tiny{\emph{GC} $=$ W}} (6,0) node[anchor=north,outer sep=3pt] {$t_{2}$};
\draw[->] (6,0) --  node[anchor=south] {$I_{3}$} node[anchor=north] {\tiny{\emph{GC} $=$ B}} (8,0);
\end{tikzpicture}
\end{center}
At the beginning, the \emph{GlobalColor} is initialized to black, and it remains black in the time interval $I_{1}$.
At the time point $t_{1}$, the \emph{GlobalColor} is set to white by a process $P_{t_{1}}$.
Then, the \emph{GlobalColor} remains as white in the time interval $I_{2}$.
Note that we are not claiming that the \emph{GlobalColor} will not be updated during $I_{2}$.
However, the \emph{GlobalColor} will not be updated to black during $I_{2}$ by the very definition of the interval $I_{2}$.
The \emph{GlobalColor} remains as white in the time interval $I_{2}$ until a process $P_{t_{2}}$ sets it to black at time point $t_{2}$.

%======================================================
%   Lemma 1
%======================================================
\begin{lem}
\label{lemma1}
If a process $P_{i}$ gets a Token.number of 2 or more, then when $P_{i}$ is calculating its Token.number in the doorway (lines 7-13), there must exist a conflicting process $P_{j}$ with the same Token.color as $P_{i}$ and a smaller Token.number.
Moreover, $P_{i}$ can not enter the CS before $P_{j}$ finished the CS.
\end{lem}

\begin{proof}
According to line 9, when a process is calculating its \emph{Token.number}, it will ignore any process with the same session or different \emph{Token.color}.
Assume there doesn't exist such a conflicting process $P_{j}$ with the same \emph{Token.color} as $P_{i}$ and smaller \emph{Token.number}.
It is easy to see that $P_{i}$ will get a \emph{Token.number} of $1$, which is a contradiction with $P_{i}$ getting a \emph{Token.number} of $2$ or more.
Therefore, such $P_{j}$ must exist when $P_{i}$ is calculating its \emph{Token.number}.

Moreover, as $P_{j}$ has the same \emph{Token.color} as $P_{i}$ and a smaller \emph{Token.number}, $P_{i}$ will wait for $P_{j}$ at line 19 until $P_{j}$ exited the CS and reset its \emph{Token}. Thus, our claim is true.
\qed
\end{proof}

%======================================================
%   Lemma 2
%======================================================

\begin{lem}
\label{lemma2}
If a process $P_{t_{i}}$ flips the GlobalColor at the time point $t_{i}$, then there must exist a conflicting process $P_{\widetilde{t_{i}}}$ with the same Token.color as $P_{t_{i}}$ and a smaller Token.number when $P_{t_{i}}$ is calculating its Token$[t_{i}]$.number (lines 7-13).
Moreover, both $P_{t_{i}}$ and $P_{\widetilde{t_{i}}}$ execute line 5 in the time interval $I_{i}$.
\end{lem}

\begin{proof}
Since process $P_{t_{i}}$ flips the \emph{GlobalColor} at time point $t_{i}$, according to line 25, $P_{t_{i}}$ must have the \emph{Token}$[t_{i}]$.\emph{number} of 2 or more.
By Lemma \ref{lemma1}, there must exist a conflicting process $P_{\widetilde{t_{i}}}$ with the same \emph{Token.color} and a smaller \emph{Token}.\emph{number} when $P_{t_{i}}$ is calculating its \emph{Token}$[t_{i}]$.\emph{number} in the doorway.
Next, we show that both $P_{t_{i}}$ and $P_{\widetilde{t_{i}}}$ execute line 5 in time interval $I_{i}$.

We prove our claim using mathematical induction. 
Without loss of generality, assume the \emph{GlobalColor} is initialized to $c$.

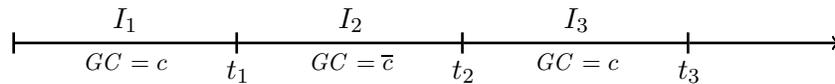
\begin{center}
\begin{tikzpicture}[line width=1pt][every node/.style={midway}]
\draw[|-|] (0,0) --  node[anchor=south] {$I_{1}$} node[anchor=north] {\small{\emph{GC} $= c$}} (3,0)  node[anchor=north,outer sep=3pt] {$t_{1}$} ;
\draw[-|] (3,0) --  node[anchor=south] {$I_{2}$} node[anchor=north] {\small{\emph{GC} $= \overline{c}$}}  (6,0) node[anchor=north,outer sep=3pt] {$t_{2}$};
\draw[-|] (6,0) --  node[anchor=south] {$I_{3}$} node[anchor=north] {\small{\emph{GC} $= c$}} (9,0) node[anchor=north,outer sep=3pt] {$t_{3}$};
\draw[->] (9,0) --  (11,0);
\end{tikzpicture}
\end{center}

\begin{description}
	\item [{$\mathbf{K=1}$}:] Process $P_{t_{1}}$ flips the \emph{GlobalColor} from $c$ to $\overline{c}$ at $t_{1}$.
	Obviously,
	$P_{t_{1}}$ and $P_{\widetilde{t_{1}}}$ must have the \emph{Token.color} of $c$.
	Since $I_{1}$ is the first and the only time interval that the \emph{GlobalColor} is $c$ before time point $t_{1}$ and so, $P_{t_{1}}$ and $P_{\widetilde{t_{1}}}$ must execute line 5 in $I_{1}$.

	\item[{$\mathbf{K=2}$}:]
	Process $P_{t_{2}}$ flips the \emph{GlobalColor} from $\overline{c}$ to $c$ at $t_{2}$.
	$P_{t_{2}}$ and $P_{\widetilde{t_{2}}}$ must have the \emph{Token.color} of $\overline{c}$.
	Since $I_{2}$ is the only time interval that the \emph{GlobalColor} is $\overline{c}$ before $t_{2}$, $P_{t_{2}}$ and $P_{\widetilde{t_{2}}}$ execute line 5 in $I_{2}$.

	\item Assume the claim holds for $K = (i-1)$.

	\item[$\mathbf{K=i}$:] Process $P_{t_{i}}$ may flip the \emph{GlobalColor} either from $c$ to $\bar{c}$ or $\bar{c}$ to $c$ at time point $t_{i}$.
	We consider only the case where \emph{GlobalColor} flips from $c$ to $\bar{c}$ at time point $t_{i}$ as the argument is similar in the other case.

	If $P_{t_{i}}$ flips the \emph{GlobalColor} from $c$ to $\bar{c}$ at $t_{i}$, it is easy to see both $P_{t_{i}}$ and $P_{\widetilde{t_{i}}}$ have the \emph{Token.color} of $c$.

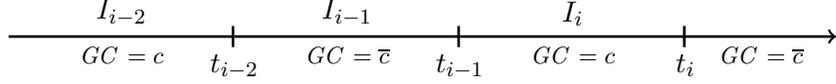
	\begin{center}
	\begin{tikzpicture}[line width=1pt][every node/.style={midway}]
	\draw[-|] (0,0) --  node[anchor=south] {$I_{i-2}$} node[anchor=north] {\small{\emph{GC} $= c$}} (3,0)  node[anchor=north,outer sep=3pt] {$t_{i-2}$} ;
	\draw[-|] (3,0) --  node[anchor=south] {$I_{i-1}$} node[anchor=north] {\small{\emph{GC} $= \overline{c}$}}  (6,0) node[anchor=north,outer sep=3pt] {$t_{i-1}$};
	\draw[-|] (6,0) --  node[anchor=south] {$I_{i}$} node[anchor=north] {\small{\emph{GC} $= c$}} (9,0) node[anchor=north,outer sep=3pt] {$t_{i}$};
	\draw[->] (9,0) -- node[anchor=north] {\small{\emph{GC} $= \overline{c}$}} (11,0);
	\end{tikzpicture}
	\end{center}

	We prove our argument by showing both $P_{t_{i}}$ and $P_{\widetilde{t_{i}}}$ execute line 5 in $I_{i}$

	\begin{enumerate}

		\item  Assume $P_{t_{i}}$ does not execute line 5 in $I_{i}$, then it must execute line 5 before $t_{i-2}$ since the \emph{GlobalColor} is $\bar{c}$ during $I_{i-1}$.
		By induction hypothesis, in time interval $I_{i-1}$, $P_{t_{i-1}}$ and $P_{\widetilde{t_{i-1}}}$ execute line 5 and then write their \emph{Token.color} of $\overline{c}$.
		Clearly, $P_{t_{i}}$ must be conflict with either $P_{t_{i-1}}$ or $P_{\widetilde{t_{i-1}}}$ since $P_{t_{i-1}}$ and $P_{\widetilde{t_{i-1}}}$ have different sessions.
		Therefore, according to line 17, at least one of $P_{t_{i-1}}$ or $P_{\widetilde{t_{i-1}}}$ will wait for $P_{t_{i}}$ to finish the doorway to set \emph{Choosing}$[t_{i}]$ to false.
		By Lemma \ref{lemma1}, $P_{t_{i-1}}$ can not enter the CS before $P_{\widetilde{t_{i-1}}}$ finished the CS.
		So, when $P_{t_{i-1}}$ enters the CS, $P_{t_{i}}$ already wrote down its \emph{Token}$[t_{i}]$.\emph{color} of $c$ and finished the doorway.
		After $P_{t_{i-1}}$ finishes the CS and checks the condition at line 26, it will find out $P_{t_{i}}$ has the opposite \emph{Token.color}, and so, $P_{t_{i-1}}$ will not update the \emph{GlobalColor}, which is contradiction with $P_{t_{i-1}}$ flips the \emph{GlobalColor} at $t_{i-1}$.
		Therefore, process $P_{t_{i}}$ executes line 5 in time interval $I_{i}$.

		\item  Since $P_{t_{i}}$ executes line 5 in $I_{i}$, it must read the \emph{Token}$[\widetilde{t_{i}}]$.\emph{number} of $P_{\widetilde{t_{i}}}$ in $I_{i}$, which implies $P_{\widetilde{t_{i}}}$ does not reset its \emph{Token}$[\widetilde{t_{i}}]$ until $P_{t_{i}}$ reads it in $I_{i}$.
		Assume $P_{\widetilde{t_{i}}}$ does not execute line 5 in $I_{i}$, then it must executes line 5 before $t_{i-2}$.
		Obviously, $P_{\widetilde{t_{i}}}$ is conflict with either $P_{t_{i-1}}$ or $P_{\widetilde{t_{i-1}}}$.
		Therefore, before $P_{t_{i-1}}$ enters the CS, $P_{\widetilde{t_{i}}}$ must already write down its \emph{Token}$[\widetilde{t_{i}}]$.\emph{color} of $c$ and finish its doorway.
		$P_{\widetilde{t_{i}}}$ keeps its \emph{Token}$[\widetilde{t_{i}}]$ at least until $P_{t_{i}}$ reads it in $I_{i}$ and this implies $P_{\widetilde{t_{i}}}$ keeps its \emph{Token}$[\widetilde{t_{i}}]$  when $P_{t_{i-1}}$ is exiting. 
		Hence, in the exit section, $P_{t_{i-1}}$ will find that $P_{\widetilde{t_{i}}}$ has the opposite \emph{Token.color} and so, $P_{t_{i-1}}$ will not flip the \emph{GlobalColor}, which is a contradiction.
		So, process $P_{\widetilde{t_{i}}}$ executes line 5 in $I_{i}$.

	\end{enumerate}

	We showed that if $P_{t_{i}}$ flips the \emph{GlobalColor} from $c$ to $\bar{c}$ at $t_{i}$, then both $P_{t_{i}}$ and $P_{\widetilde{t_{i}}}$ execute line 5 in $I_{i}$.

\end{description}

Hence, we proved our claim for the case $K=i$ and therefore the result follows 
by mathematical induction.
\qed
\end{proof}

%======================================================
%   Lemma 3
%======================================================

\begin{lem}
\label{lemma3}
After a process $P_{i}$ executes line 5, the GlobalColor can not be flipped more than once before $P_{i}$ finished the exit section (line 34).
\end{lem}

\begin{proof}

Without loss of generality, we assume $P_{i}$ reads the \emph{GlobalColor} of $c$ at line 5 in the time interval $I_{i}$.
Assume the \emph{GlobalColor} is flipped twice at time $t_{i}$ and $t_{i+1}$ before $P_{i}$ finished the exit section.
$P_{t_{i+1}}$ represents the process that flips the \emph{GlobalColor} at $t_{i+1}$.

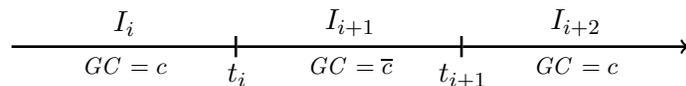
\begin{center}
\begin{tikzpicture}[line width=1pt][every node/.style={midway}]
\draw[-|] (0,0) --  node[anchor=south] {$I_{i}$} node[anchor=north] {\small{\emph{GC} $= c$}} (3,0)  node[anchor=north,outer sep=3pt] {$t_{i}$} ;
\draw[-|] (3,0) --  node[anchor=south] {$I_{i+1}$} node[anchor=north] {\small{\emph{GC} $= \overline{c}$}}  (6,0) node[anchor=north,outer sep=3pt] {$t_{i+1}$};
\draw[->] (6,0) --  node[anchor=south] {$I_{i+2}$} node[anchor=north] {\small{\emph{GC} $= c$}} (9,0) ;
\end{tikzpicture}
\end{center}

By Lemma \ref{lemma1} and Lemma \ref{lemma2}, in $I_{i+1}$, two conflicting processes $P_{t_{i+1}}$ and $P_{\widetilde{t_{i+1}}}$ execute line 5 and then finish the CS.
So, $P_{i}$ must be conflict at least with either $P_{t_{i+1}}$ or $P_{\widetilde{t_{i+1}}}$.
Before $P_{t_{i+1}}$ enters the CS, $P_{i}$ must finish the doorway and write down its \emph{Token}$[i]$.\emph{color} of $c$.
Thus, after $P_{t_{i+1}}$ finished the CS, it will find out $P_{i}$ has the opposite \emph{Token.color} and fail to update the \emph{GlobalColor}, which is a contradiction with our assumption that $P_{t_{i+1}}$ flips the \emph{GlobalColor} at $t_{i+1}$.
Hence, we proved the theorem.
\qed
\end{proof}

%======================================================
%  Lemma 4
%======================================================
\begin{lem}
\label{lemma4}
If process $P_{i}$ and process $P_{j}$ request conflicting sessions and process $P_{i}$ finished the doorway before process $P_{j}$ executes line 4, then process $P_{j}$ does not enter the CS before process $P_{i}$ finished the CS.
\end{lem}

\begin{proof}

For the sake of concreteness, assume process $P_{i}$ gets the \emph{Token.color} of $c$.
Process $P_{j}$ may have a \emph{Token.color} of either $c$ or $\bar{c}$.
We consider both possibilities.
\begin{enumerate}

	\item If process $P_{j}$ has the \emph{Token.color} of $c$, it will get a larger \emph{Token.number} than $P_{i}$ as $P_{j}$ has not begun to calculate its \emph{Token.number} when $P_{i}$ finished the doorway. 
	Thus $P_{j}$ will wait for $P_{i}$ at line 19 until $P_{i}$ finished the CS and resets its \emph{Token} at line 34.

	\item If process $P_{j}$ gets the \emph{Token.color} of $\bar{c}$, that means the \emph{GlobalColor} was flipped to $\bar{c}$ after $P_{i}$ read the \emph{GlobalColor} of $c$ at line 5. 
	By Lemma \ref{lemma3}, the \emph{GlobalColor} can not be flipped again to $c$ until $P_{i}$ finished the exit section.
	Hence, $P_{j}$ will find out it has the different \emph{Token.color} than $P_{i}$ and waits for it at line 21. None of the conditions in line 21 can become true until $P_{i}$ finished the CS and then reset its \emph{Token}$[i]$.
\end{enumerate}
We showed that in all cases, process $P_{j}$ can not enter the CS before process $P_{i}$ finished the CS.
\qed
\end{proof}

%======================================================
%   FCFS
%======================================================
\begin{thm}
\label{fcfs2}
The BWBGME algorithm satisfies the FCFS property.
\end{thm}

\begin{proof}
If processes $P_{i}$ and $P_{j}$ request conflicting sessions and $P_{i}$ finished the doorway before $P_{j}$ starts the doorway, by Lemma \ref{lemma4}, $P_{j}$ can not enter the CS before process $P_{i}$ entered the CS.
\qed
\end{proof}

%=======================================================
%   Mutual Exclusion
%=======================================================
\begin{thm}
The BWBGME algorithm satisfies the Mutual Exclusion property.
\end{thm}

\begin{proof}
Suppose the algorithm does not satisfy the mutual exclusion property. 
At some point of time, there exist two processes $P_{i}$ and $P_{j}$ requesting different sessions are in the CS simultaneously.
By Lemma \ref{lemma4}, neither $P_{i}$ nor $P_{j}$ can finish their doorway before the other executes line 4.
Thus, processes $P_{i}$ and $P_{j}$ must have an overlap when they execute lines 4-15. 
Since $P_{i}$ and $P_{j}$ have different sessions, they wait for each other at line 17 to finish the doorway to set \emph{Choosing} to false before determining which one of them has the priority.

For the sake of concreteness, we assume that process $P_{i}$ enters the CS with the \emph{Token.color} of $c$.
When $P_{i}$ checks line 18 for $P_{j}$, it may have the \emph{Token}$[j]$.\emph{color} of $c$ or $\bar{c}$.
We analyze both cases. 
\begin{enumerate}

    \item Process $P_{i}$ finds that $P_{j}$ has the same \emph{Token.color} of $c$.
    In order for $P_{i}$ to enter the CS, it must have a smaller \emph{Token}.\emph{number} in comparison to $P_{j}$. 
    So, process $P_{j}$ can not enter the CS before $P_{i}$ finished the CS and reset its \emph{Token}$[i]$ as it has the larger \emph{Token.number}, contradicting the assumption.

    \item Process $P_{i}$ finds that process $P_{j}$ has the opposite \emph{Token}.\emph{color} of $\bar{c}$.
    There are two possibilities: (I) after process $P_{j}$ reads the \emph{GlobalColor} of $\overline{c}$ at line 5, the \emph{GlobalColor} is flipped to $c$ and then process $P_{i}$ executes line 5 get the \emph{Token}$[i]$.\emph{color} of $c$ or (II) after process $P_{i}$ reads the \emph{GlobalColor} at line 5 of $c$, it is flipped to $\overline{c}$ and then process $P_{j}$ read it at line 5 and got the \emph{Token}$[j]$.\emph{color} of $\bar{c}$.
    In order for process $P_{i}$ enter the CS, it must find out the \emph{GlobalColor} is $\overline{c}$ when it checks line 21.
    Therefore, by Lemma \ref{lemma3}, only case II can be true.
    On the other hand, when process $P_{j}$ executes line 21, it will find that none of the conditions are satisfied because  the $\emph{GlobalColor}$ can not be flipped to $c$ again by Lemma \ref{lemma3}.
    Therefore, $P_{j}$ will wait for $P_{i}$ until it finishes the exit section, contradicting the assumption that $P_{j}$ and $P_{i}$ stays in the CS at the same time.
\end{enumerate}

We proved that in all cases we get a contradiction with our assumption, and hence, the algorithm satisfies the Mutual Exclusion property.
\qed
\end{proof}

%=====================================================
%   Bounded Exit
%====================================================
\begin{thm}
The BWBGME algorithm satisfies the Bounded Exit property.
\end{thm}

\begin{proof}
Since the exit section does not contain any busy-wait loop, a process that enters the exit section will finish it within a bounded number of its own steps.
\qed
\end{proof}

%=====================================================
%   Concurrent Entry
%=====================================================
\begin{thm}
The BWBGME algorithm satisfies the Concurrent Entry property.
\end{thm}

\begin{proof}
If a process $P_{i}$ requests a session, and no other process requests a different session, that means for every other process $P_{j}$, it holds that \emph{Token}$[j]$.\emph{session}$\in\{0, $\emph{mysession}\} from the view of $P_{i}$.
Since process $P_{i}$ always checks \emph{Token}$[j]$.\emph{session} $\in\{0, $\emph{mysession}\} in busy-wait lines 19 and 21, it will not wait on these two lines. 
If a process $P_{j}$ has the session of 0, it is in the remainder section and so \emph{Choosing}$[j]$ is false. 
Thus, in line 17 process $P_{i}$ can not be blocked by any such process because it will find that either $\emph{Token}[j].\emph{session} =  \emph{mysession}$ or $\emph{Choosing}[j] = $ false. 
Hence, $P_{i}$ can not wait at any line when there is no conflict and so, it enters the CS within a bounded number of its own steps.
\qed
\end{proof}

%======================================================
%   Deadlock Freedom
%======================================================
\begin{thm}
\label{deadlock}
 The BWBGME algorithm satisfies the Deadlock Freedom property.
\end{thm}

\begin{proof}
Suppose the algorithm does not satisfy the deadlock freedom property.
Then there is an execution of the algorithm in which a nonempty set $S$ of processes that enter the entry section but none of them enters the CS, and no process enters the CS infinitely often. 
As every process in the entry section will finish the doorway eventually and set the \emph{Choosing} to be false, processes in set $S$ can not wait on line 17 forever.
Let $S_{1}$ be the subset of $S$ that consists of processes having the \emph{Token.color} of $c$, and $S_{2}$ be the subset of $S$ that consists of processes having the \emph{Token.color} of $\bar{c}$. 
Since no process enters the CS infinitely often, the \emph{GlobalColor} can not be changed infinitely often .
Assume the \emph{GlobalColor} remains as $c$ after some time.

We observe that processes in $S_{2}$ can not wait on line 21 forever, because any process in $S_{2}$ will find that their \emph{Token.color} is different with the \emph{GlobalColor} and stop waiting. 
So, processes in $S_{2}$ can only wait on line 19. 
Also, we observe that processes in $S_{2}$ can not wait for any process in $S_{1}$ on line 19, because processes in $S_{2}$ will find that they have different \emph{Token.color} with processes in $S_{1}$ at line 19, which will make processes in $S_{2}$ immediately terminate waiting. 
Thus, a process in $S_{2}$ must wait for another process in $S_{2}$ at line 19.
In the set $S_{2}$, there must be a process that has the smallest \emph{Token.number} and that process will pass all other processes in $S_{2}$ at line 19 and enter the CS.
This contradicts with our assumption that processes in $S$ can not enter the CS.
So $S_{2} = \emptyset$.

On the other hand, a process in $S_{1}$ can not wait for any other process in $S_{1}$ at line 21 because it will find they have the the same \emph{Token.color} and stop waiting.
Therefore, every process in $S_{1}$ must wait for another process in $S_{1}$ at line 19.
In the set $S_{1}$, there must exist a process that has the smallest \emph{Token.number} among all processes in $S_{1}$ and that process will pass all other processes in $S_{1}$ and enter the CS, which is a contradiction.
So $S_{1} = \emptyset$.

Hence, we have $S = S_{1} \cup S_{2} =  \emptyset$, which contradicts with our assumption and the theorem is proved.
\qed
\end{proof}

%======================================================
%   Starvation Freedom
%======================================================
\begin{thm}
The BWBGME algorithm satisfies the Starvation Freedom property.
\end{thm}

\begin{proof}
Lamport proved that \cite{Lamp86} if an algorithm satisfies the deadlock freedom property and the FCFS property, then it necessarily satisfies the starvation freedom property.
Hence, combining Theorem \ref{fcfs2} and Theorem \ref{deadlock}, we know the algorithm has the starvation freedom property.
\qed
\end{proof}

%======================================================
%   Bounded Register
%======================================================
\begin{thm}
\label{bwbound}
The shared variables used in the BWBGME algorithm are bounded.
\end{thm}

\begin{proof}

Obviously, every shared variable used in the algorithm is bounded except possibly \emph{Token.number} and \emph{Token.session}.
However, if at all the \emph{Token.session} is unbounded,
it is due to the nature of application and not due to the design of the algorithm. 
In other words, if the underlying application uses only a bounded number of sessions,
so will our algorithm.  Hence, we need to concern ourselves with only \emph{Token.number}.
Here we show that the value of \emph{Token.number} can not be larger than $N+1$.

If not, there must exist a execution sequence

\begin{equation}
\label{number}
P_{i_{1}}{@}5 \rightarrow P_{i_{1}}{@}13 \rightarrow P_{i_{2}}{@}13 \rightarrow P_{i_{3}}{@}13 \rightarrow ... \rightarrow  P_{i_{j}}{@}13 \rightarrow...  \rightarrow P_{i_{N+2}}{@}13
\end{equation}

\noindent such that all processes in the sequence have the same \emph{Token.color} of $c$, and for all $j$ except $1$, $P_{i_{j}}$ gets its \emph{Token}$[i_{j}]$.\emph{number} at line 13 by incrementing the previous largest number \emph{Token}$[i_{j-1}]$.\emph{number}.
$P_{i_{1}}$ gets its \emph{Token.number} of $1$ by virtue of it not finding any conflicting process with the same \emph{Token.color} when it executes lines 7-12.
It is easy to see that for any $j$, $P_{i_{j}}$ and $P_{i_{j+1}}$ have different session numbers as $P_{i_{j+1}}$ increments $P_{i_{j}}$'s \emph{Token.number}. 
Here we prove the theorem by showing that such a sequence doesn't exist.

We first show that every process in the sequence \ref{number} executes line 5 in the same time interval $I$.

If not, there exists two processes in the sequence $P_{i_{k}}$ and $P_{i_{k+1}}$ executes line 5 in different interval.
As they have the same \emph{Token.color}, after one of them executes line 5, the \emph{GlobalColor} is flipped twice before the other process executes line 5.
By Lemma \ref{lemma3}, the previous process already resets its \emph{Token} before the \emph{GlobalColor} is flipped twice, and therefore, $P_{i_{k+1}}$ will fail to read $P_{i_{k}}$'s \emph{Token.number} and increment the number, which contradicts our assumption that every process in the sequence \ref{number} increments the previous largest \emph{Token.number}.
So, every process in the sequence \ref{number} executes line 5 in the same interval $I$.

Since we only have $N$ processes, by pigeon hole principle, at least two 
processes, say $P_{i_{x}}$ and $P_{i_{y}}$ (not necessarily distinct processes, but with $x < y$),
from the above sequence ($P_{i_{1}}, P_{i_{2}}, ..,P_{i_{N+2}}$) must exit their CS before $P_{i_{N+2}}$ executes line number 3. So, we have
\[ P_{i_{x}}{@}34 \rightarrow  P_{i_{y}}{@}34 \rightarrow P_{i_{N+2}}{@}3 \]
It is easy to see that every two adjacent processes in sequence \ref{number} have different sessions, hence, for any $j$, $P_{i_{j+1}}$ can not enter the CS before $P_{i_{j}}$ completely exits and reset its \emph{Token}$[i_{j}]$.  Hence, we can conclude that $P_{i_{1}}$ and $P_{i_{2}}$ must have exited by
the time $P_{i_{y}}$ exits (as $P_{i_{1}}$ and $P_{i_{2}}$ are the first two processes 
in the sequence \ref{number}, $P_{i_{x}}$ and $P_{i_{y}}$ are some processes in the sequence \ref{number} and
the processes in the sequence \ref{number} exit one after the other).
Thus, before $P_{i_{N+2}}$ begins to execute the doorway, $P_{i_{1}}$ and $P_{i_{2}}$ must already have finished their exit section, which is represented by the following sequence.
\[ P_{i_{1}}{@}34 \rightarrow P_{i_{2}}{@}34 \rightarrow P_{i_{N+2}}{@}3 \]
By our previous claim, $P_{i_{1}}$, $P_{i_{2}} \ldots P_{i_{N+2}}$ execute line 5 in
the same interval and so they all get the same color, say $c$. 
In order for $P_{i_{1}}$ to enter the CS, all processes conflicting with $P_{i_{1}}$ in the system, having \emph{Token.color} of $\bar{c}$ must have already finished the exit section and 
reset their \emph{Token}.
In order for $P_{i_{2}}$ to enter the CS, all processes conflicting with $P_{i_{2}}$ in the system, having \emph{Token.color} of $\bar{c}$ must have already finished the exit section and 
reset their \emph{Token}.
Every process that has the \emph{Token.color} of $\bar{c}$ must be in
conflict with either $P_{i_{1}}$ or $P_{i_{2}}$, as these two processes have
different session numbers. 
Hence, after $P_{i_{2}}$ finishes the CS and checks the condition at line 26, it will find that there is no process with the opposite \emph{Token.color} of $\bar{c}$ is in the system.
So, $P_{i_{2}}$ will flip the \emph{GlobalColor} to $\overline{c}$, which immediately starts a new time interval $I'$.
This contradicts our previous claim  that every process in the sequence \ref{number} executes line 5 in the same interval $I$ (as $P_{i_{N+2}}$ in particular will execute line 5 in a different interval).

So, we showed that such a sequence of processes can not exist and so proved that the value of the shared variable $\emph{Token.number}$ is bounded by $N+1$.
\qed
\end{proof}

%======================================================
%   Complexity Analysis
%======================================================

%==============================================
%   RMR complexity in CC model
%==============================================
\begin{thm}
\label{bwcomp}
The BWBGME Algorithm has $O(N)$ shared space complexity and $O(N)$ RMR Complexity under the CC model.
\end{thm}

\begin{proof}

It is trivial to observe that algorithm is of $O(N)$ shared space complexity.
Therefore, we focus on proving the RMR complexity here.
From line 17 through line 22, there are three busy-wait loops, one each at lines 
17, 19 and 21. Since line 19 and line 21 are located in different branches, it is not difficult to see that a process that goes through lines 17-22 executes only two busy-wait loops viz., line 17 and line 19 or line 17 and line 21.

In line 17, when a process $P_{i}$ is busy waiting for a process $P_{j}$, if \emph{Choosing}$[j]$ changes to false, then process $P_{i}$ will immediately terminate the wait.
It is possible that process $P_{j}$ executes a new invocation and resets the \emph{Choosing}$[j]$ to true again with another conflicting session.
However, in that case, since process $P_{i}$ doorway proceeds $P_{j}$, according the FCFS property, $P_{j}$ can not reenter the CS before $P_{i}$.
So, this can occur at most once.
Thus, line 17 can only involve at most five RMR (three for \emph{Choosing}$[j]$ and two for \emph{Token}$[j]$).

Suppose process $P_{i}$ in the entry section executes line 19 and waits for a process $P_{j}$.
If process $P_{j}$ changes its \emph{Token}$[j]$, it will be either resetting its \emph{Token}, or getting a new \emph{Token} in a new invocation.
If process $P_{j}$ resets its \emph{Token}$[j]$, $P_{i}$ will immediately terminate the wait because it will detect \emph{Token}$[j]$.\emph{session} is $0$.
However, it is possible that $P_{j}$ wants to reenter the CS with a new \emph{Token} before $P_{i}$ detects that $P_{j}$ is in the remainder section.
In view of the FCFS property, process $P_{j}$ can not reenter the CS before process $P_{i}$ enters the CS.
So, sooner or later process $P_{i}$ will be able to compare its \emph{Token}$[i]$ with that of process $P_{j}$.
Therefore, the maximum number of RMR at line 12 is two (for \emph{Token}$[j]$).

Suppose process $P_{i}$ in the entry section executes line 21 and waits for another process $P_{j}$.
Without loss of generality, assume $P_{i}$ gets the \emph{Token.color} of $c$.
If the \emph{GlobalColor} is set to $\bar{c}$, then process $P_{i}$ will find out \emph{GlobalColor} $\neq $ \emph{mycolor} and stop the wait.
Note that \emph{GlobalColor} can not be changed twice to $c$ again when process $P_{i}$ is waiting at line 21 in view of Lemma \ref{lemma3}.
Hence, process $P_{i}$ will be able to detect the fact that the \emph{GlobalColor} has changed sooner or later. 

It is important to realize that if the variable \emph{GlobalColor} is 
overwritten, this will immediately invalidate the cached copies of it, even if the actual 
value of the \emph{GlobalColor} did not really change. This could potentially
cause multiple RMR when a process $P_{i}$ is waiting in line 21.
Suppose that the \emph{GlobalColor} is set again to $c$ by another process $P_{k}$ with the \emph{Token}$[k]$.\emph{color} of $\bar{c}$ executing the exit section. Now, this will cause
$P_{i}$ to migrate the \emph{GlobalColor} to cache again, but  $P_{i}$ will continue
to be blocked as the value of \emph{GlobalColor} did not really change.
In that case, by Lemma \ref{lemma3}, $P_{k}$ can not start a new invocation and set the \emph{GlobalColor} to $c$ again before $P_{i}$ enters the CS.
This is because, if $P_{k}$ gets its \emph{Token.color} of $\bar{c}$ in the new invocation, the \emph{GlobalColor} is actually flipped to $\bar{c}$ after $P_{i}$ executes line 5. 
When $P_{k}$ finished the CS in the new invocation, it will find out $P_{i}$ has the opposite \emph{Token}$[i]$.\emph{color} of $c$ and fail to set the \emph{GlobalColor}.
So, $P_{i}$ will see that the \emph{GlobalColor} is updated, but still get blocked at line 21 for at most $N-1$ times in the whole loop lines 16-23.
Hence, the amortized RMR complexity of \emph{GlobalColor} for each process $P_{j}$ is constant. 

For other conditions at line 21, if process $P_{j}$ resets the \emph{Token}$[j]$, $P_{i}$ will be able to pass process $P_{j}$.
It is possible that $P_{j}$ may want to reenter the CS before $P_{i}$ detected that $P_{j}$ is in the remainder section.
In the new invocation, process $P_{j}$ may execute line 3 and write \emph{Token}$[j]$ of $($\emph{mysession} $, \perp, $\ $0)$ and block $P_{i}$ again.
However, by the FCFS property, process $P_{j}$ can not reenter the CS before process $P_{i}$ and so $P_{i}$ will be able to detect the new \emph{Token}$[j]$ or that the \emph{GlobalColor} is flipped when process $P_{j}$ finished the doorway, which makes process $P_{i}$ stop waiting.  Hence, the overall RMR made at line 21 is constant.  

Since the RMR complexity of lines 17, 19 and 21 is constant.
and lines 17-22 are enclosed within a {\em for loop} that can run a maximum of $N$ times, it follows that the waiting room (lines 16-23) is of $O(N)$ RMR complexity.

It is trivial to see that lines 8-11 costs constant RMR and so, lines 7-12 is of $O(N)$ RMR complexity.
It is not difficult to see that the function used in the algorithm is of $O(N)$ RMR complexity.
It is also seen that other lines involve constant RMR.
Hence, the overall RMR complexity of this algorithm is $O(N)$ in the CC model.
\qed
\end{proof}

\end{document}